\documentclass[12pts,a4paper]{article}
\usepackage{amssymb,amsthm,amsmath,amsfonts}
\usepackage{setspace}
\usepackage{graphicx}
\usepackage{subfigure}
\usepackage{titling}
\usepackage[margin=1in]{geometry}
\usepackage[nottoc]{tocbibind}
\usepackage[spaces,hyphens]{url}
\newtheorem{theorem}{Theorem}[section]
\newtheorem{defn}{Definition}[section]
\newtheorem{ bremark}{Remark}[section]
\title{ More on  Boundary Behavior of Univalent Harmonic Mappings} 
\author{
Gebreslassie Atsbha Weldegebrial\textsuperscript{1$\ast$}
Hunduma Legesse Geleta\textsuperscript{2} \\
\\
\textsuperscript{1$\ast$}mathematics, Addis Ababa University, Arada, Addis Ababa , 1176, Addis Ababa, Ethiopia \\
\textsuperscript{2}mathematics, Addis baba University, Arada, Addis Ababa , 1176, Addis Ababa, Ethiopia \\
$\ast$Corresponding author: \texttt{gebreslassie.atsbha@aau.edu.et}   https://orcid.org/0009-0004-6409-1236\\
Contributing author: \texttt{hunduma.legesse@aau.edu.et} \\
\dag~These authors contributed and approved this work.
}

\date{\today}

\begin{document}

\maketitle

\pagenumbering{arabic}

%\subjclass[2000]{Primary 11M41}
 %\thanks{}

%\date{\today}          % Enter your date or \today between curly braces

\textbf{ Abstract}\\

\textit{Many authors have examined various boundary behaviors of injective harmonic mappings in the open unit disk. Building on Laugesen's work, Bshouty and others explored the boundary behavior of harmonic mappings under different conditions. In this paper, we extend their work and find out the angular limits of the arguments and logarithms of analytic functions under various conditions. We also examined the dilatation possesses only a  finite set  of zeros within any stolz angle if the first derivative of harmonic function $f$ at the boundary is positive infinity.}\\

\textbf{Keywords/phrases}: Univalent harmonic function; Angular limit; Dilatation; Boundary Behavior.
\section{Introduction}

Throughout this article, GRM, $\mathbb{C}$, $\Omega$, $\partial\Omega,$   $\Delta$, $\mathbb{T}$ , $\omega$ and  $\Phi$ denote General Riemann  Mapping theorem,the complex plane, region in the complex plane, boundary of the region, the open unit disk, boundary of the unit disk, the dilatation and boundary function  respectively.\\

 Given that  $\Phi: \partial\Omega\rightarrow\mathbb{C}$ to be  continuous. Then there exist a continuous mapping $f: \overline{\Omega}\rightarrow\mathbb{C}$ in which $\Phi(z) = f(z)$ for all $z\in \partial\Omega$ and $f$ is harmonic in $\Omega$~\textbf{\cite{conway2012functions}}. Then the Poisson integral $f(z) = \frac{1}{2\pi}\int_{-\pi}^{\pi}P(r,\varphi-\theta)\Phi(e^{i\varphi})d\varphi$, $z=re^{i\theta}\in\Delta$, where $F(r,\theta)$ is the Poisson kernel of $\Delta$,  is the solution for the Dirichlet problem. Furthermore, $P(r,\varphi-\theta$ it is harmonic mapping  in $\Delta$ ~\cite{bshouty2010problems}.\\

The study of harmonic mapping is basically the extension  of analytic functions. These mappings have many applications in different aspects of physics and other domains where La place's equation is prominent. To this effect, research on harmonic functions is motivated both by their mathematical contribution and their practical relationships to physical settings.\textbf{\cite{laugesen1997planar}}.\\

%5The investigation  of univalent harmonic mappings gained focus from complex analysts following the %prestigious paper by Clunie and Sheil-Small~\cite{sheil1989fourier} in 1984. In 1986, Hengartner and %Schober~\cite{hengartner1986boundary} made efforts to reformulate an appropriate version of the %Riemann mapping theorem for harmonic mappings, drawing on the theory of quasiconformal mappings. The %investigation of these scholars with their colleagues, led to the emergence of different open %problems, conjectures, and unresolved questions. While some of these conjectures have been solved, %lots of  challenging questions demands further research. \\

A recent area of attention is the boundary values of univalent complex-valued harmonic mappings.
%Laugesen tried to determined conditions on the boundary behavior of 
%$f$ where its dilation is an infinite Blaschke product. If $f$ has infinitely piecewise  constant boundary values (\& maps onto a convex infinite sided polygon) then $\omega$\ is an infinite Blachke product \textbf{\cite{laugesen1997planar}}.
 Bshouty et al \textbf{\cite{bshouty2012boundar}}, studied the angular limits of arg $h'(z)$ and arg $g'(z)$ provided $\frac{d\Phi}{d\theta}(e^{i\theta}) = 0$ where,  $f$ is the Riemann mapping from $\Delta$ onto a bounded convex region. In this study, we further extend on  these works by finding different conditions on the boundary behavior of  $f$ to determine the angular limit of analytic functions provided that  $\frac{d\Phi}{d\theta}(e^{i\theta})= 0.$	We calculate  the angular limits of log $h'(z)$ , log $g'(z)$, argument of $ h'(z)$ and argument of  $g'(z)$  given that  $\frac{d\Phi}{d\theta}(e^{i\theta}) \neq 0.$ \\  Another purpose of this paper is to examine a problem posed by A.Lyzzaik  et al \textbf{\cite[problem 3.19]{bshouty2010problems}} . This problem is solved partially that if  $e^{i\theta_{0}}$ exists, then $\omega$ possesses only a  finite set  of zeros within any stolz angle at $e^{i\theta_{0}}$ for  $|\frac{d\Phi}{d\theta}(e^{i\theta_{0}})|\leq$c, where  c is a specific constant . Building up on this, we addresses the dilatation $\omega$ exhibits finitely many roots in each stolz angle $S_{\theta_{0}}$with vertex $e^{i\theta_{0}}$ if $\frac{d\Phi}{d\theta}(e^{i\theta_{0}})$ is $+\infty.$\\
% stated "$\frac{df^*}{d\theta}(e^{i\theta}) $ exists, is it true that $\omega$ has finitely many zeros in any stolz angle at $e^{i\theta}$? "\\
The structure of this paper is as follows: In part 2, we review some foundational results that will be important for proving the key findings. In Section 3, we state and prove the primary results of the paper. Theorem 3.1 formulate  constraints on the boundary values of a complex-valued harmonic function 
$f$ to determine the angular limits of the logarithm of analytic functions. Theorems 3.2 and 3.3 provide expressions for the logarithm and argument of analytic functions, assuming the derivative of 
$f$ at the boundary is nonzero finite.
Theorem 3.4 finds  the dilatation $\omega$ contains  only a  finite set  of zeros within  any stolz angle at $e^{i\theta_{0}}$ if $\frac{d\Phi}{d\theta}(e^{i\theta_{0}})$ is $+\infty.$
 Furthermore, we have proved that there is no interior zeros of $f'$ in a stolz angle at $e^{i\theta_{0}}$ if GRM maps $\Delta$ onto bounded convex domain.
%\section{Preliminaries}
In this part, we present key ideas and findings that will be essential for proving the main results. We start by presenting some classic outcomes, along with significant definitions and theorems.\\

%\textbf{A harmonic mapping  $f$  } of a complex region $G$ is complex-valued function that satisfies Laplace's equation $\Delta$$f$ $\equiv$$f_{xx}+f_{yy}=0.$ Any complex-valued harmonic function can be written in the form $f(z)$=$h(z)$+$\overline{g(z)},$ where $z = x+iy$ and  $h$ \& $g$ are both analytic function.\\
%\begin{defn}\cite{nguyen2020linear}.
%A one to one complex valued harmonic function is said to be univalent harmonic function.   %\textit{Locally univalent}  functions are functions that are one to one locally.
%\end{defn}
\begin{defn}\cite{romney2013class}.
A complex valued harmonic functions $f$ = $h+\overline{g}$ is said to be orientation preserving at $z_{0}$ if $J_{f}(z_{0})$ $> 0$ and is orientation reversing at $z_{0}$ if $J_{f}(z_{0})$ $< 0$, where  $J_{f}(z_{0})$ is the jacobian of $f$ which is formulated as\\
   $$J_{f}(z_{0})= |f_{z}|^2-|f_{\overline{z}}|^2 = |h'|^2-|g'|^2$$

 Its dilatation $w$ is given by $\omega(z) =  \frac{g'(z)}{h'(z)}.$ 
% The dilatation may be interpreted geometrically to %represent the "stretching" or "distortion " of $f$. %Dilatation $\omega$ measures how far $f$ is from being %conformal. The mapping $f$ is conformal if and only if %$\omega$ = 0.
\end{defn}

%\begin{defn}\cite{nguyen2020linear}.
%A region in the complex plane is considered convex if, for any two points within the region, the line segment joining them lies entirely within the region.//
% It is convex in the direction 
%$\theta$if any line segment, parallel to the ray 
 %$re^{i\theta}$ connecting two points in the region, %remains entirely within the region.
%\end{defn}
%\begin{theorem}\cite{nguyen2020linear}. Every function convex in one  direction is univalent.

%\end{theorem}

\begin{defn}\cite{Privalov1950}.
 A Stolz angle is a sector in the upper half-plane with its vertex on the x-axis, touching the x-axis only at that single point. For $e^{i\theta_{0}}\in \mathbb{T},$\\
 $$ S_{\alpha}(e^{i\theta_{0}}) = \left\{z\in \Delta:|z - e^{i\theta_{0}}|<\alpha(1 - |z|)\right\}, 0<\alpha<\infty$$
\end{defn}
%\begin{defn}\cite{laugesen1997planar}. A Blaschke product is an analytic function of the form $$B(z)= %e^{i\theta}\prod_{j=1}^{k}\frac{-\overline{z_j}}{|z_j|}\frac{z-z_j}{1-z\overline{z_j}}$$ for $|z|<1,$ %where $\theta \in\mathbb{R},$ $1\leq k \leq\infty,$ and $|z_j|<1$ for all $j,$  and $\sum_{j=1}^{k}%1-|z_j|<\infty$ with the convention that $\frac{-\overline{z_j}}{|z_j|}$=1 , if $z_j$=0. When %$k<\infty,$ we call $\textbf{B}$  a finite \textbf{Blaschke} product of degree $k,$ and when 5$k=\infty,$ we call an infinite $\textbf{Blaschke}$ product.
%\end{defn}

  \begin{theorem}
  (Hengartner and Schober \cite{hengartner1986boundary}) Suppose $\Omega$ be a bounded  connected open set Having a boundary that is locally connected.  Assume  that $\omega(\Delta)\subset\Delta$ and $\omega_{0}$ is an invariant point of $\Omega.$ Then for $\bar{f_{\bar{z}}} = \omega(z)f_{{z}}$ there exists a solution satisfying:\\
(a) $f(0)$ = $\omega_{0}$, $f_{z}(0)>0$, and $f(\Delta)\subset\Omega$\\
(b) There exist  a countable set $A\subset\mathbb{T}$ for which the unrestricted limits $\Phi(e^{it})$ = $\lim_{z\to e^{it}}f(z)$ exist on $\mathbb{T}$$\backslash$A and belongs to $\partial\Omega.$\\
(c) $\Phi(e^{it-})$  = ess$lim_{s\uparrow t}\Phi(e^{is})$ and $\Phi(e^{it+})$  = ess$lim_{s\downarrow t}\Phi(e^{is})$\\
exist on $\mathbb{T}$, is on the boundary of  $\Omega$ and are the same on $\mathbb{T}$$\backslash$A.\\
(d) The accumulation points of $f$ at $e^{it}\in$ A is the direct line segment connecting $\Phi(e^{it-})$ to  $\Phi(e^{it+}).$
  \end{theorem}
  Note:  $f$ is known as  GRM from $\Delta$ onto $\Omega.$ 
  \begin{theorem}\cite{bshouty2012boundar}.
  Let $\Phi$ be  real-valued function on  $\mathbb{T}$ that is lebesgue integrable. Then\\
  (i) If $(\frac{d\Phi}{d\theta})(e^{i\theta_{0}})$ exist and is finite, then
 $$
  lim_{z\rightarrow e^{i\theta}}\frac{\partial f}{\partial \theta}(z) = \frac{d\Phi}{d\theta}(e^{i\theta_{0}})
  $$
  exhibits  uniform behavior in any stolz angle with vertex located at $e^{i\theta_{0}}.$\\
   (ii) If $(\frac{d\Phi}{d\theta})(e^{i\theta_{0}})$ = $+\infty,$ then 
 $$
  lim_{z\rightarrow r^{}-1}\frac{\partial f}{\partial \theta}(re^{i\theta}) = +\infty
  $$
  %If $\Phi$ is  monotone increasing in a neighborhood of $\theta_{0}$, then 
% $$ lim_{z\rightarrow r^{}-1}\frac{\partial f}{\partial \theta}(re^{i\theta}) = +\infty$$
   exhibit uniform behavior in any stolz sector with vertex located at $e^{i\theta_{0}}.$
  \end{theorem}
  \begin{theorem}\cite{bshouty2012boundar}.
  Assume  $f$ denotes  an injective  harmonic mapping of $\Delta$ onto a convex region $\Omega$ containing the origin  and let $f(0)$ = 0. Then $|f_{z}|^2$ + $|f_{\bar{z}}|^2 \geq\frac{dist(0,\partial\Omega)^2}{16}.$
  \end{theorem}
 % The following Lemma is stated and proved in 
  %\cite[corollary 1]{bshouty2012boundar}.
 % \begin{cor}
 %  Let $f = h + \overline{g}$ be the GRM from the unit disk $\Delta$ onto a %bounded convex set with boundary function $f^*$ and dilatation $\omega.$ If  $%%(\frac{df^*}{d\theta})(e^{i\theta_{0}})$ = 0 for some $\theta_{0} \in\mathbb{R}$ %and the angular limit , $lim_{z\rightarrow e^{i\theta_{0}}}$arg $\omega(z)$ = $\alpha\neq0$, then the angular limits\\
%(a)   $lim_{z\rightarrow e^{i\theta_{0}}}$ arg $h'(z)$ = $ -\theta_{0} - \frac{1}{2}\alpha$  $(\mod\pi)$ and \\
%(b)  $lim_{z\rightarrow e^{i\theta_{0}}}$ arg $g'(z)$ = $ -\theta_{0} +\frac{1}%{2}\alpha$ $(\mod\pi)$ hold.
   
 % \end{cor}

\section{Main Results}
\begin{theorem}
 Suppose  $f = h + \overline{g}$ denotes  the Riemann Mapping from $\Delta$ onto a bounded convex set.  If $(\frac{d\Phi}{d\theta})(e^{i\theta_{0}})$ = 0 for some $\theta_{0} \in\mathbb{R}.$ Then\\
  I) If  $lim_{z\rightarrow e^{i\theta_{0}}}$arg $\omega(z)$ = $\alpha\neq0$, then\\
 (a) $lim_{z\rightarrow e^{i\theta_{0}}}$ log $h'(z)$ = $-i(\theta_{0} + \frac{\alpha}{2}) + \frac{log|\omega|}{2}$  $\pmod\pi$ and \\
(b)  $lim_{z\rightarrow e^{i\theta_{0}}}$ log $g'(z)$ = $i(\frac{\alpha}{2} - \theta_{0}) + \frac{3log|\omega|}{2}$ $\pmod\pi)$.\\

 II) If $lim_{z\rightarrow e^{i\theta_{0}}}$arg $\omega(z)$ = 0, then\\
(c) $lim_{z\rightarrow e^{i\theta_{0}}}$ log $h'(z)$ = $-i\theta_{0} + \frac{log|\omega|}{2}$  $\pmod\pi$ and \\
(d)  $lim_{z\rightarrow e^{i\theta_{0}}}$ log $g'(z)$ = $-i\theta_{0} + \frac{3log|\omega|}{2}$ $\pmod\pi$.
\end{theorem} 
  \begin{proof} (a).  We know from  theorem 2.2(i) above that if $(\frac{d\Phi}{d\theta})(e^{i\theta_{0}})$ exists and is finite. Thus
 $$
  lim_{z\rightarrow e^{i\theta}}\frac{\partial f}{\partial \theta}(z) = \frac{d\Phi}{d\theta}(e^{i\theta_{0}})
  $$
  exhibits  uniform behavior  in each stolz angle with vertex at $e^{i\theta_{0}}.$ \\
   From the Poisson formula, we derive 
   $$f(z) = \frac{1}{2\pi} \int_{0}^{2\pi}\frac{1-|z|^2}{|e^{i\theta}-z|^2}\Phi(e^{i\theta})d\theta =  \frac{1}{2\pi}\int_{0}^{2\pi}\frac{e^{i\theta}}{e^{i\theta}-z}\Phi(e^{i\theta)}d\theta +  \frac{1}{2\pi}\overline{\int_{0}^{2\pi}\frac{z}{e^{i\theta}-z}\overline{\Phi(e^{i\theta)}d\theta}}$$
$$ = h(z)+\overline{g(z)}$$
Differentiating with respect to $\theta$ under the integral sign (which is justifiable), we obtain

$$f_{\theta}(z) = h_{\theta}(z) + \overline{g_{\theta}(z)} = izh'(z)+\overline{izg'(z)}$$It follows that $\frac{d\Phi}{d\theta}(e^{i\theta_{0}}) = \frac{\partial f}{\partial \theta}(z)$ = $i(zh'(z)) - \overline{zg'(z)}.$\\
  But by theorem 2.3,  $h'$ is different from 0 and consequently, then \\
 
  -i $lim_{z\rightarrow e^{i\theta}}\frac{\partial f}{\partial \theta}(z)$ = $lim_{z\rightarrow e^{\theta_{0}}}zh'(1 - \frac{\overline{zh'}}{zh'}\overline{\omega})$ =  0 \\
  
  From this, we obtain  $lim_{z\rightarrow e^{i\theta_{0}}} (1 - \frac{\overline{zh'}}{zh'}\overline{\omega}) = 0.$\\
Therefore,
  $$lim_{z\rightarrow e^{i\theta_{0}}} log(1) = lim_{z\rightarrow e^{i\theta_{0}}} log (\overline{zh'\omega}) - lim_{z\rightarrow e^{i\theta_{0}}} log (zh').$$
  This gives the desired result.\\
  
   (b) It is  know that $lim_{z\rightarrow e^{i\theta_{0}}}$ log $\omega$ = log $|\omega|$ + iarg $\omega$ =  log $|\omega|$ + i$\alpha.$ 
 Then $$lim_{z\rightarrow e^{i\theta_{0}}} log( g'(z)) =  lim_{z\rightarrow e^{i\theta_{0}}} log h'(z) +  lim_{z\rightarrow e^{i\theta_{0}}} log \omega(z).$$ Which is the desired result.\\
 
 c and d follows, since $lim_{z\rightarrow e^{i\theta_{0}}}$arg $\omega(z)$ = 0 = $\alpha.$

   \end{proof}

\begin{theorem}
Assume  $f = h + \overline{g}$ denotes the General Riemann Mapping from the open unit disk  onto a bounded convex set. If $(\frac{d\Phi}{d\theta})(e^{i\theta_{0}}) = \gamma\neq0,\infty$ for some $\theta_{0} \in\mathbb{R},$\\

 I) If $lim_{z\rightarrow e^{i\theta_{0}}}$arg $\omega(z)$ = $\alpha\neq0$, then \\
 (a) $lim_{z\rightarrow e^{i\theta_{0}}}$ arg $h'(z) =  \frac{\pi}{4} - \frac{\alpha}{2} - \theta_{0}\pmod\pi$ and \\
 (b) $lim_{z\rightarrow e^{i\theta_{0}}}$ arg $g'(z)$ = $\frac{\pi}{4} + \frac{\alpha}{2} - \theta_{0}~ \pmod\pi$.\\
 
  II) If $lim_{z\rightarrow e^{i\theta_{0}}}$arg $\omega(z)$ = 0, then\\
(c) $lim_{z\rightarrow e^{i\theta_{0}}}$ arg $g'(z)$ = $lim_{z\rightarrow e^{i\theta_{0}}}$ arg $h'(z) =  \frac{\pi}{4}  - \theta_{0} \pmod\pi$
 \end{theorem}
\begin{proof} (a)
Since  $(\frac{d\Phi}{d\theta})(e^{i\theta_{0}}) = \gamma\neq0,\infty$ exists and  finite , from theorem 2.2(i) above  the angular limit 
 $$
  lim_{z\rightarrow e^{i\theta}}\frac{\partial f}{\partial \theta}(z) = \frac{d\Phi}{d\theta}(e^{i\theta_{0}})
  .$$ 
  It follows that \\
  
 i $lim_{z\rightarrow e^{i\theta}}\frac{\partial f}{\partial \theta}(z)$ = $lim_{z\rightarrow e^{\theta_{0}}}zh'(1 - \frac{\overline{zh'}}{zh'}\overline{\omega}) =  i\gamma$ which is different from zero. \\
 
  Consequently, $h'\neq0.$ \\
 
 Thus,  $lim_{z\rightarrow e^{i\theta_{0}}}$ arg $(i\gamma) = lim_{z\rightarrow e^{i\theta}} (arg z + arg h' - arg \overline{z} - arg \overline{h'} - arg \overline{\omega} ).$\\
  Which implies\\
 $$\theta_{0} + 2arg h' + \theta_{0} + \alpha = \frac{\pi}{2}~ (\mod \pi).$$
Which is as desired.\\

(b) We know that the angular limit  $$lim_{z\rightarrow e^{i\theta_{0}}} arg g'(z) =  lim_{z\rightarrow e^{i\theta_{0}}} arg h'(z) +  lim_{z\rightarrow e^{i\theta_{0}}} arg \omega(z)$$  $$= lim_{z\rightarrow e^{i\theta_{0}}} arg g'(z) + \alpha.$$  Which is also as desired.\\

Similarly, plugging  $\alpha = 0,$ (c) follows directly.\\

This completes the proof.
\end{proof} 
\begin{theorem}
Suppose $f = h + \overline{g}$ denotes the General Riemann Mapping from the open unit disk  onto a bounded convex set. If $(\frac{d\Phi}{d\theta})(e^{i\theta_{0}}) = \gamma\neq0,\infty$ for some $\theta_{0} \in\mathbb{R},$\\
 I) If $lim_{z\rightarrow e^{i\theta_{0}}}$arg $\omega(z)$ = $\alpha$, then \\
 (a) $lim_{z\rightarrow e^{i\theta_{0}}}$ log $h'(z) = \frac{1}{2}(log|\frac{\gamma}{\omega}| +\frac{\pi}{2} - 2i\theta_{0} - i\alpha) \pmod\pi$ and \\
 (b) $lim_{z\rightarrow e^{i\theta_{0}}}$ log $g'(z)$ = $\frac{1}{2}(log |\gamma\omega| + i\alpha) + \frac{\pi}{4} - i\theta_{0} \pmod\pi$.\\
 
 II)  If $lim_{z\rightarrow e^{i\theta_{0}}}$arg $\omega(z)$ = 0, then\\
 (c) $lim_{z\rightarrow e^{i\theta_{0}}}$ log $h'(z) = \frac{1}{2}(log|\frac{\gamma}{\omega}| +\frac{\pi}{2} - 2i\theta_{0}) \pmod\pi$ and \\
 (d) $lim_{z\rightarrow e^{i\theta_{0}}}$ log $g'(z)$ = $\frac{1}{2}(log |\gamma\omega) + \frac{\pi}{4} - i\theta_{0} \pmod\pi$ 
 \end{theorem}
 
 \begin{proof} (a)
Since  $(\frac{d\Phi}{d\theta})(e^{i\theta_{0}}) = \gamma\neq0,\infty$ exists and  finite , from theorem 2.2(i) above  the following holds true. 
 $$
  lim_{z\rightarrow e^{i\theta}}\frac{\partial f}{\partial \theta}(z) = \frac{d\Phi}{d\theta}(e^{i\theta_{0}})
  .$$ 
  It follows that \\
  $$-ilim_{z\rightarrow e^{i\theta}}\frac{\partial f}{\partial \theta}(z) = lim_{z\rightarrow e^{\theta_{0}}}zh'(1 - \frac{\overline{zh'}}{zh'}\overline{\omega}) =  i\gamma \neq 0.$$ 
   Consequently, $h'\neq 0$ and we have the following, \\
 
  $$lim_{z\rightarrow e^{i\theta_{0}}} log (i\gamma) = lim_{z\rightarrow e^{i\theta}} (log z + log h' - log \overline{z} - log \overline{h'} - log \overline{\omega} ).$$\\
 From this we obtain,\\
  $$ 2i\theta_{0} +2log h' + log |\omega| +  i\alpha - log |\gamma| = \frac{\pi}{2} \pmod \pi.$$
  Hence proved.\\
  
  (b) We know that the angular limit  $$lim_{z\rightarrow e^{i\theta_{0}}} arg g'(z) =  lim_{z\rightarrow e^{i\theta_{0}}} log h'(z) +  lim_{z\rightarrow e^{i\theta_{0}}} log \omega(z) = lim_{z\rightarrow e^{i\theta_{0}}} log h'(z) + log |\omega| + i\alpha$$ holds. Hence the result follows.\\
  
  Since $lim_{z\rightarrow e^{i\theta_{0}}}$arg $\omega(z)$ = 0 = $\alpha,$ (c) and (d) follows directly.
\end{proof}
\begin{theorem}
Assume  $f = h + \overline{g}$ denotes the GRM from  $\Delta$ onto $\Delta.$ If  $\frac{d\Phi}{d\theta}(e^{i\theta_{0}})$ is $+\infty.$ Then $\omega$ contains a finite number of zeros in any given stolz angle $S_{\theta_{0}}$ at $e^{i\theta_{0}}.$
\begin{proof}
We know from  theorem 2.3 (ii) above if $(\frac{d\Phi}{d\theta})(e^{i\theta_{0}})$ = $+\infty$, then 
$$
  lim_{z\rightarrow r^{}-1}\frac{\partial f}{\partial \theta}(re^{i\theta}) = +\infty
  $$ along any non tangential path.
  Since 
  $$ f = h+\overline{g}$$   
   $$ = \frac{1}{2\pi} \int_{0}^{2\pi}\frac{1-|z|^2}{|e^{i\theta}-z|^2}\Phi(e^{i\theta})d\theta$$
     $$ =\frac{1}{2\pi}\int_{0}^{2\pi}\frac{e^{i\theta}}{e^{i\theta}-z}\Phi(e^{i\theta)}d\theta +  \frac{1}{2\pi}\overline{\int_{0}^{2\pi}\frac{z}{e^{i\theta}-z}\overline{\Phi(e^{i\theta)}d\theta}}$$

Differentiating with respect to $\theta$, we get\\
   $$\frac{d\Phi}{d\theta}(e^{i\theta_{0}}) = lim_{z\rightarrow e^{i\theta}}\frac{\partial f}{\partial \theta}(z)$$  $$ = ilim_{z\rightarrow e^{i\theta}}(zh'(z)) - \overline{zg'(z)}.$$\\
   Dividing both  sides by $\bar{z} = e^{-i\theta},$ we get \\
  $$lim_{z\rightarrow e^{i\theta}}|e^{2i\theta}h'(z) - \overline{g'(z)}| =  +\infty$$\\
  This step proceeds to\\
   $$\overline{lim}_{z\rightarrow e^{i\theta}}(|h'(z)| - |g'(z)|)\leq\infty$$\\
   so that $$ - \infty+ \underline{lim}_{z\rightarrow e^{i\theta}} |h'(z)|\leq \underline{lim}_{z\rightarrow e^{i\theta}}|g'(z)|$$\\
   Thus  for any finite value of $|h'(z)|$, $lim_{z\rightarrow e^{i\theta}}|g'(z)|>0.$\\
    Consequently, we can observe  that the accumulation points of $g'(z)$ at $e^{i\theta_{0}}$  doesn't contain zero produces at once that $g'(z)$  contains only a  finite set  of zeros within  any stolz angle $S_{\theta_{0}}$ at $e^{i\theta_{0}}.$ Since $\omega = \frac{g'}{h'},$ finite number of these zeros are zeros of  $\omega.$\\
 Moreover, By theorem 2.3 above $h'(z)$ is different from zero. Consequently, $\omega$ exhibit finite number of zeros in any stolz sector $S_{\theta_{0}}$ at $e^{i\theta_{0}}.$\\
  To simplify the argument, we apply another mechanism using privalov boundary uniqueness\cite{Privalov1950}. We know that, 
    $$f_{\theta}(z) = izh'(z) - i\overline{zg'(z)}$$
    From this $||h'(z)|-|g'(z)|| \leq |h'(z)| + |g'(z)|.$  Rearranging we have $$1\leq \frac{|f_{\theta}|}{h(1-|\omega|)}\leq \frac{1+|\omega|}{1-|\omega|}$$
    If $\frac{d\Phi}{d\theta}(e^{i\theta_{0}}) = \infty,$ then $|f_{\theta}(z)| \rightarrow \infty$ as $r\rightarrow1$
    From the inequality, we observe $|\omega(re^{i\theta_{0}})| \rightarrow1.$
    Thus, $\omega$ has a boundary (non-tangential) limit of modulus 1 at $e^{i\theta_{0}}$
    $$lim_{z\rightarrow e^{i\theta_{0}}} = e^{i\alpha}, |e^{i\alpha}| = 1\neq0$$
    But by privalov boundary uniqueness\cite{Privalov1950}, a bounded analytic function tends to a non-zero limit at a boundary point can't have infinitely many zeros accumulating non tangentially at that point. Therefore,$\omega$  can have only a finite number of zeros in any stolz angle $e^{i\theta_{0}}$, which gives the desired result.
  \end{proof}
\end{theorem}
\begin{theorem} 
  Let $f = h + \overline{g}$ be a univalent harmonic mapping from $\Delta$ on to a bounded convex domain $\Omega.$ Suppose there exist a continuous differential angular function $\theta$ such that $z(r) = re^{i\theta(r)}\in \Delta$  and argf = $\theta_{0}$ is constant. Then
 
  \[\theta(r) = \left \{ \begin{array}{ll}
  \frac{-a(logr+1)}{r}+ \frac{c}{r}-\theta_{0}, & \mbox{if argf' is constant}\\
 -\phi+\sin^{-1}(\frac{c}{r}),& \mbox{otherwise}
  \end{array}\right.\], where a and c are constants.
  \begin{proof}
  Let $z(r)=re^{i\theta(r)}$ with $r\rightarrow1^{-}.$
  Using the identity $$ \frac{d}{dr}argf' = Im\left(\frac{f'}{f}\right)$$.
  Along the path: $$\frac{d}{dr}f(z(r)) = f'(z(r))\frac{dz}{dr} = f'(z)\left(e^{i\theta(r)} + ir\theta'(r)e^{i\theta(r)}\right).$$
  But since argf = $\theta_{0}$ is constant , the derivative must vanish.
  $Im\left( Fe^{i\theta(r)}(i + ir\theta'(r))\right)$ = 0, where $F = \frac{f'}{f} = Re^{i\phi},$ 
  so, the imaginary part becomes
  $RIm\left(e^{i(\theta+\phi)}(1+ir\theta'(r))\right).$\\
  We require the imaginary part to be zero. After computing and expanding we have: 
 $$\theta'(r) = - \frac{tan(\phi+\theta(r))}{r}$$\\
  Suppose argF = $\theta_{0}+\frac{a}{r}$. Let $\Psi = \phi + \theta(r)$.
  Consequently, we have:
 $$\Phi'(r)+ \frac{tan\Psi}{r} = -\frac{a}{r^2}$$For small value of $\Phi$, using Taylor series approximation tan$\Phi\approx\Phi$, this reduces to
  $\frac{d}{dr}(r\Psi) = -\frac{a}{r}$. Integrating  both side  yields the desired expression for $\theta(r)$.
 Thus,  $\theta(r)$ describes a spiral path inside the unit disk whose angle depends logarithmically and rationally on r.\\
  If argF is constant , then 
  $$\theta'(r) = \Phi'(r) = -\frac{tan\Psi}{r}$$ . After arranging and integrating, the result follows immediately.
 
  \end{proof}
  
  \end{theorem}
  \begin{theorem}
  Assume  $f = h + \overline{g}$ denotes the General Riemann Mapping from the open unit disk  onto a bounded convex domain $\Omega$.
  If $lim_{z\rightarrow e^{i\theta_{0}}}$arg $\omega(z)$ = $\beta$, then \\
  a. $lim_{z\rightarrow e^{i\theta_{0}}}$arg $f'(z)$ = $-\frac{\beta}{2}\pmod\pi$ \\
  b). If $lim_{z\rightarrow e^{i\theta_{0}}}|\omega(z)| = \lambda,$ where $z\in S_{\theta_{0}}(\alpha)$  $\&$   $ \lambda\in[0, 1).$\\
  Then:  $$0\leq\frac{m}{2}\leq|h'(z)|\leq \frac{M}{1-\lambda}$$
  where, m = inf$|f_{z}(z)|>0$ and M = sup$|f_{z}(z)|$\\
  c).There is no  interior zeros of $f'$ in a stolz angle at $e^{i\theta_{0}}$.\\
  
  \begin{proof}
  a. Given that $f = h + \overline{g}$.  $$f'(z) = h'(z) + \overline{h'(z)\omega(z)}$$.
  Let $h'(z) = R(z)e^{i\phi(z)})$,   $\omega(z) = r(z)e^{i\theta(z).}$
  Then: $$f'(z) = R(z)e^{i\phi(z)}\left[1 + r(z)e^{-i(\theta(z)+2\phi(z))}\right]$$.
  This implies: 
  $$lim_{z\rightarrow e^{i\theta_{0}}}argf'(z) = \alpha + arg\left[1 + e^{-i(2\alpha+\beta)}\right]$$, Assuming that $lim_{z\rightarrow e^{i\theta_{0}}}$arg $h'(z)$ = $\alpha.$\\
  Set  $2\alpha+\beta = \zeta$. Since $\frac{sin\zeta}{1+cos\zeta}$ = $tan\frac{\zeta}{2},$  then arg$\left[1+e^{-i\zeta}\right]$ = $-\frac{\zeta}{2}.$\\
      
  Consequently, $lim_{z\rightarrow e^{i\theta_{0}}}$arg $f'(z)$ = $\alpha - \left(\frac{2\alpha+\beta}{2}\right)$ =    $-\frac{\beta}{2}$\\
  Even though both $h'(z)$ and $\omega$ contribute to $f$, the limiting angle of  the full derivative $f'(z)$ near the boundary depends only on arg$\omega.$ This shows  that the harmonic shear(through $\omega)$ completely controls the direction of the tangent vector to the image curve near $e^{i\theta_{0}}.$\\
  \textbf{Example}: Consider the univalent harmonic function
  $$f(z) = h + \overline{g} = Re\left[\frac{i}{2}log \left(\frac{i+z}{i-z}\right)\right] + i Im \left[ \frac{1}{2}log \left(\frac{i+z}{i-z}\right)\right]$$
  With the dilatation  $\omega(z) = \frac{g'(z)}{h'(z)} = z^2.$\\
  
  Since  $lim_{z\rightarrow e^{i\theta_{0}}}$arg $ \omega(z)$ = -2 $lim_{z\rightarrow e^{i\theta_{0}}}$arg$z\pmod\pi$, Then from the above proof (a) we have 
  $lim_{z\rightarrow e^{i\theta_{0}}}$arg $f'(z)$ = -2arg$z \pmod\pi$\\
 $f$ maps  $\Delta$ univalently and sense- preservingly $|\omega(z)|<1).$ The mobius transform $\frac{i+z}{i-z}$ sends $\Delta$ to the upper half-plane, and the logarithm further it to a vertical strip.The combination of real and imaginary parts rotates and scales this strip to produce convex square region as shown in figure 1 below. The derivative satisfies $lim_{z\rightarrow e^{i\theta_{0}}}$arg $ \omega(z)$ = -2 $lim_{z\rightarrow e^{i\theta_{0}}}$arg$z\pmod\pi$, ensuring convex boundary with out folding or self intersection.
   
   \newpage
  \begin{figure}
\centering
  \includegraphics[scale=0.5]{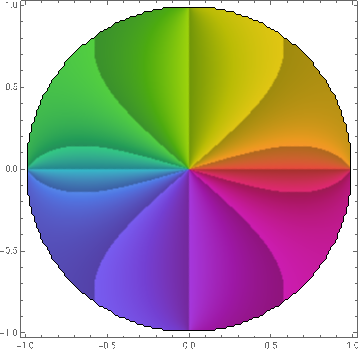}
  \caption{Image of f($\Delta$)}
  \end{figure}
  
  (b) Since, $h'(z) = \frac{f_{z}(z) - \omega(z)\overline{f_{\bar{z}}}}{1-|\omega|^2}$\\
   $$|h'(z)|\geq \frac{|f_{z}(z)| + |\omega(z){f_{\bar{z}}}|}{1 -|\omega|^2}\geq\frac{m}{2}$$
   Similarly, $|h'(z)|\leq\frac{M}{1-\lambda}.$
   The result follows immediately.
  
  For (c) we know that $f'(z) = h'(z)\left(1+\frac{\overline{\omega h'}}{h'}\right).$\\
  Consequently, $f'(z) = 0.$ This leads to -1 = $\frac{\overline{\omega h'}}{h'}.$ Since f is sense preserving $h'(z)\neq0.$\\
  1 = $|\omega(z)|.$ but f is GRM and $|\omega(z)|<1, for |z|<1.$  It contradicts. \\
  Therefore, there is no zeros of f' inside a stolz angle.\\
 Example:
 consider  $$f(z) = z + \frac{\overline{kz^2}}{2}, 0<k<1$$.\\
 We have, $\omega(z) = kz,$ $|\omega(z)|<1,$ $0<k<1$, $|z|<1$, showing that f is sense preserving. \\
  $$f'(z) = 1+k\overline{z}= 0.$$ This implies $|z| = \frac{1}{k}>1.$ so no zeros in the $\Delta.$ Hence no zeros in any  stolz angle.\\
   Since there are no zeros of $f'$ in $\Delta,$ no folding occurs in fig 2. Light grey mesh(grid inside the image) are lines of the deformed grid. They come from vertical and horizontal lines of the disk being mapped by f. The red boundary curve is the image of the unit circle. It tells us the outer shape of the image domain. Furthermore, the image is simply connected, convex domain with out cusps or self intersection.
 \begin{figure}[h]
  \centering
  \begin{minipage}{0.49\textwidth}
  \centering
  \includegraphics[width=0.5\textwidth]{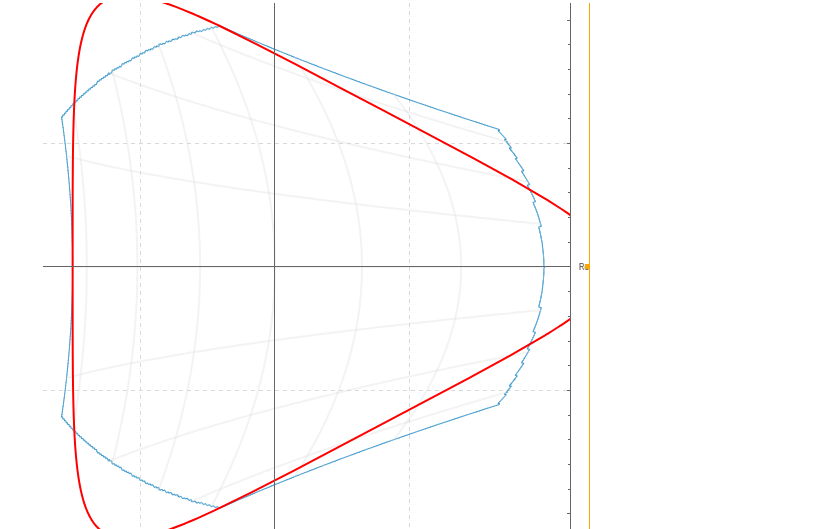}
  \caption{Image of $f(\Delta),  k = 0.5$}

  \end{minipage}
   \hfill
    \begin{minipage}{0.5\textwidth}
  \centering
  \includegraphics[width=0.99\textwidth]{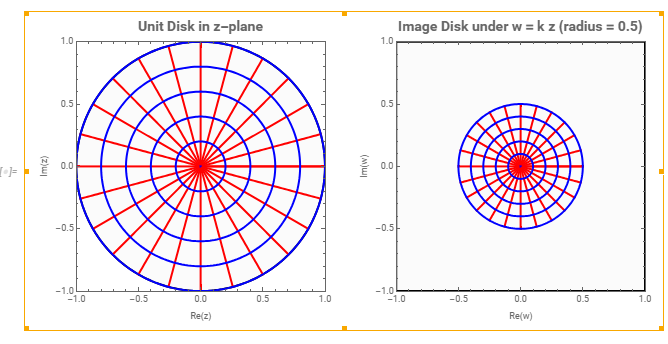}
  \caption{Image of $\Delta$ under $\omega(z) = kz, k = 0.5$}

  \end{minipage}
  
  \end{figure}
 
  \end{proof}
   
\end{theorem}

 \section*{Conclusion}                           
 In this paper, we  set various conditions on the boundary behavior of  complex-valued injective harmonic mapping of $f$ to determine  the angular limits of the argument and logarithm of analytic functions. We also studied the angular limits of the arguments and logarithms of analytic functions provided the angular limit of the dilatation is constant. , We have examined that  the dilatation $\omega$ possesses finite number of zeros in any stolz angle if the first derivative of $f$ at the boundary is positive infinity. Furthermore, We have shown that $f'$ has no  interior zeros with in any  stolz angle at $e^{i\theta_{0}}$ provided that  $f$ from the open unit disk  onto a bounded convex domain \

 %\textbf{Acknowledgment}

\textbf{Conflict  of interest declaration}\\
The authors disclose no conflict of interest in connection with this paper.
%\addcontentsline{toc}{chapter}{Biliography}
%\bibliographystyle{plain}
%\bibliography{gev}

 \end{document}